\documentclass[english,american,aps,prl,reprint]{revtex4-1}
\usepackage{lmodern}

\usepackage[T1]{fontenc}
\usepackage[latin9]{inputenc}
\usepackage{babel}
\usepackage{amsthm}
\usepackage{amsmath}
\usepackage{amssymb}
\usepackage[unicode=true,pdfusetitle,
 bookmarks=true,bookmarksnumbered=false,bookmarksopen=false,
 breaklinks=false,pdfborder={0 0 0},backref=false,colorlinks=false]
 {hyperref}
\usepackage{breakurl}

\makeatletter
\@ifundefined{textcolor}{}
{%
 \definecolor{BLACK}{gray}{0}
 \definecolor{WHITE}{gray}{1}
 \definecolor{RED}{rgb}{1,0,0}
 \definecolor{GREEN}{rgb}{0,1,0}
 \definecolor{BLUE}{rgb}{0,0,1}
 \definecolor{CYAN}{cmyk}{1,0,0,0}
 \definecolor{MAGENTA}{cmyk}{0,1,0,0}
 \definecolor{YELLOW}{cmyk}{0,0,1,0}
 }
\theoremstyle{plain}
\newtheorem{thm}{\protect\theoremname}

\usepackage{lmodern}
\usepackage{braket}
\usepackage[usenames,dvipsnames]
{pstricks}\usepackage{epsfig}

\makeatother

  \addto\captionsamerican{\renewcommand{\theoremname}{Theorem}}
  \addto\captionsenglish{\renewcommand{\theoremname}{Theorem}}
\providecommand{\theoremname}{Theorem}

\begin{document}

\title{Behavior of Quantum Correlations under Local Noise}

\author{Alexander Streltsov}

\email{streltsov@thphy.uni-duesseldorf.de}

\author{Hermann Kampermann}

\author{Dagmar Bruß}

\affiliation{Heinrich-Heine-Universität Düsseldorf, Institut für Theoretische
Physik III, D-40225 Düsseldorf, Germany}
\begin{abstract}
We characterize the behavior of quantum correlations under the influence
of local noisy channels. Intuition suggests that such noise should
be detrimental for quantumness. When considering qubit systems, we
show for which channel this is indeed the case: the amount of quantum
correlations can only decrease under the action of unital channels.
However, non-unital channels (e.g. such as dissipation) can create
quantum correlations for some initially classical state. Furthermore,
for higher-dimensional systems even unital channels may increase the
amount of quantum correlations. Thus, counterintuitively, local decoherence
can generate quantum correlations.
\end{abstract}
\maketitle
Composite quantum states often reveal puzzling features of nature.
Recently, much interest \cite{Merali2011} has been devoted to the
study of quantum correlations that may arise without entanglement:
here, the quantumness of a composite system manifests itself even
in a separable state. The fact that such quantum correlations are
present \cite{Datta2008} in an algorithm for mixed state quantum
computing \cite{Knill1998} has stimulated intensive investigations
into measures for quantum correlations \cite{Ollivier2001,Henderson2001,Oppenheim2002,Horodecki2005a,Modi2010,Daki'c2010,Adesso2010,Giorda2010,Streltsov2011,Piani2011}
and their properties and interpretations \cite{Zurek2003a,Vedral2003,Koashi2004,Lanyon2008,Werlang2010,Ferraro2010,Lang2010,Fanchini,Cavalcanti2011,Madhok2011,Rodr'iguez-Rosario2011,Luo2011,Gharibian}.
Some studies of the dynamics of quantum correlations have been presented
in \cite{Shabani2009,Mazzola2010}.

An appeal of mixed state quantum computation lies in the possibility
to be run in a noisy environment: pure entangled states are typically
fragile, and the resource of entanglement is easily destroyed by noise.
For an open system the transition from entangled to separable states
is only a matter of time - as the volume of the set of separable states
is non-zero \cite{.Zyczkowski1998}, typically it takes a finite time
for entanglement to disappear under noise such as dissipation or decoherence
\cite{Zurek2003}.

Mixed state quantum computation as suggested in \cite{Knill1998}
already uses separable states, so it is natural to assume that it
can be run also in a noisy environment. However, in order to verify
or falsify this conjecture, one has to study the behavior of quantum
correlations under noisy channels (described by trace-preserving completely
positive maps). Here we only consider {\em local} noisy channels
- as correlated channels may also preserve entanglement (with or even
without some degradation, depending on the amount of correlation),
see e.g. \cite{Macchiavello2002}. The goal of this Letter is to answer
questions such as: Which types of noisy channels decrease the amount
of quantum correlations? Are there any noisy channels that might even
{\em increase} the amount of quantum correlations? How does dissipation
influence quantum correlations, and how are they affected by decoherence?
- We point out that our answers to these questions also apply to the
situation where one actively performs local operations on a composite
quantum system, e.g. with the aim of creating or preserving quantum
correlations.

In general, a bipartite quantum state is called fully \emph{classically
correlated}, if it can be written in the form \cite{Oppenheim2002,Horodecki2005a}
\begin{equation}
\rho_{cc}=\sum_{i,j}p_{ij}\ket{i^{A}}\bra{i^{A}}\otimes\ket{j^{B}}\bra{j^{B}},\label{eq:cc}
\end{equation}
where $\{\ket{i^{A}}\}$ and $\{\ket{j^{B}}\}$ are sets of orthogonal
states of party A and B, respectively, with nonnegative probabilities
$p_{ij}$ that add up to one. If a state cannot be written as in Eq.
(\ref{eq:cc}), it is called \emph{quantum correlated. }These definitions
can be extended to any number of parties \cite{Piani2011}. As a simple
example consider the classically correlated state of two qubits 
\begin{equation}
\rho_{cc}=\frac{1}{2}\ket{0^{A}}\bra{0^{A}}\otimes\ket{0^{B}}\bra{0^{B}}+\frac{1}{2}\ket{1^{A}}\bra{1^{A}}\otimes\ket{1^{B}}\bra{1^{B}}.\label{example}
\end{equation}
Using a local channel on qubit $A$ only (namely a local measurement
and subsequent replacement) it is possible to create from the classically
correlated state (\ref{example}) the quantum correlated state 
\begin{equation}
\rho=\frac{1}{2}\ket{0^{A}}\bra{0^{A}}\otimes\ket{0^{B}}\bra{0^{B}}+\frac{1}{2}\ket{+^{A}}\bra{+^{A}}\otimes\ket{1^{B}}\bra{1^{B}}\label{outexample}
\end{equation}
 with $\ket{+^{A}}=\frac{1}{\sqrt{2}}\left(\ket{0}+\ket{1}\right)$.
The quantum channel that achieves this transformation can be formally
written as the completely positive trace-preserving map 
\begin{equation}
\rho=\Lambda_{A}\left(\rho_{cc}\right)=E_{1}\rho E_{1}^{\dagger}+E_{2}\rho E_{2}^{\dagger}\label{cpmap}
\end{equation}
 with local Kraus operators $E_{1}=\ket{0^{A}}\bra{0^{A}}$ and $E_{2}=\ket{+^{A}}\bra{1^{A}}$
acting only on qubit $A$. The state in Eq. (\ref{outexample}) is
not of the form (\ref{eq:cc}), i.e. it is quantum correlated.

As will become clear below in this Letter, one reason why the local
quantum channel in Eq. (\ref{cpmap}) is able to create quantum correlations
lies in its action on the maximally mixed state $\frac{1}{2}\openone_{A}$.
Observe that $\Lambda_{A}\left(\frac{1}{2}\openone_{A}\right)=\frac{1}{2}\ket{0^{A}}\bra{0^{A}}+\frac{1}{2}\ket{+^{A}}\bra{+^{A}}\neq\frac{1}{2}\openone_{A}$.
This property is also known as \emph{non-unitality}. A single-qubit
quantum channel $\Lambda$ is called unital if and only if it maps
the maximally mixed state onto itself: $\Lambda\left(\frac{1}{2}\openone\right)=\frac{1}{2}\openone$,
see also Fig. \ref{fig:unital}. We will turn this observation into
Theorem \ref{thm:1} by showing that non-unitality is one property
which enables a local channel to create quantum correlations in a
multi-qubit system. In Theorem \ref{thm:2} we will show that on the
other hand local unital quantum channels cannot increase quantum correlations
in a multi-qubit system. However, this statement does not hold for
higher dimensions. 
\begin{figure}
\scalebox{1} 
{
\definecolor{color}{rgb}{0.35294117647058826,0.592156862745098,1.0}
\definecolor{color2}{rgb}{0.058823529411764705,0.9333333333333333,0.42745098039215684}
\definecolor{color3}{rgb}{0.9568627450980393,0.8431372549019608,0.054901960784313725}

\begin{pspicture}(-1.5,-2)(6.5,5.5)
\psellipse[linewidth=0.04,dimen=outer,fillstyle=solid,fillcolor=color](0,0)(1.5,1.5)
\psellipse[linewidth=0.04,dimen=outer,linestyle=dashed](0,0)(1.5,0.6)
\psellipse[linewidth=0.04,dimen=outer,fillstyle=solid,fillcolor=color](5,0)(1.5,1.5)
\psellipse[linewidth=0.04,dimen=outer,linestyle=dashed](5,0)(1.5,0.6)
\psdot(0,0)
\psdot(5,0)

\psellipse[linewidth=0.04,dimen=outer,fillstyle=solid,fillcolor=color](0,4)(1.5,1.5)
\psellipse[linewidth=0.04,dimen=outer,linestyle=dashed](0,4)(1.5,0.6)
\psellipse[linewidth=0.04,dimen=outer,fillstyle=solid,fillcolor=color](5,4)(1.5,1.5)
\psellipse[linewidth=0.04,dimen=outer,linestyle=dashed](5,4)(1.5,0.6)
\psdot(0,4)
\psdot(5,4)

\rput{-8}(0,0){
\psline[linewidth=0.04cm](0,0)(3,0.27)
\psline[linewidth=0.04cm](0,0)(3,-0.27)
\psellipse[linewidth=0.04,dimen=outer,fillstyle=solid,fillcolor=color3](3,0)(0.2,0.29)
\psline[linewidth=0.5cm,linecolor=color3,arrowsize=0.1cm 2.0,arrowlength=1.5,arrowinset=0.4]{->}(3,0)(5,0)
}

\rput{0}(0,4){
\psline[linewidth=0.04cm](0,0)(3,0.27)
\psline[linewidth=0.04cm](0,0)(3,-0.27)
\psellipse[linewidth=0.04,dimen=outer,fillstyle=solid,fillcolor=color2](3,0)(0.2,0.29)
\psline[linewidth=0.5cm,linecolor=color2,arrowsize=0.1cm 2.0,arrowlength=1.5,arrowinset=0.4]{->}(3,0)(5,0)
\psdot(5,0)}

\rput(-0.3,0){$\frac{1}{2}\openone$}
\rput(5.3,0){$\frac{1}{2}\openone$}
\rput(-0.3,4){$\frac{1}{2}\openone$}
\rput(5.3,4){$\frac{1}{2}\openone$}
\rput(4,4){$\Lambda_u$}

\rput{70}(0,0){
\psline[linewidth=0.04cm,arrowsize=0.1cm 2.0,arrowlength=1.5,arrowinset=0.4]{<->}(-1.5,0)(1.5,0)
\psdot(-1.5,0)
\psdot(1.5,0)}
\rput(-0.5,-1.65){$\ket{\psi_1}$}
\rput(0.6,1.65){$\ket{\psi_2}$}

\rput(4.95,-0.9){$\sigma$}
\rput(4.4,-0.28){$\sigma_1$}
\rput(5.55,-1.1){$\sigma_2$}
\rput(4.0,-0.55){$\Lambda_{nu}$}

\rput{70}(0,4){
\psline[linewidth=0.04cm,arrowsize=0.1cm 2.0,arrowlength=1.5,arrowinset=0.4]{<->}(-1.5,0)(1.5,0)
\psdot(-1.5,0)
\psdot(1.5,0)}
\rput(-0.5,2.35){$\ket{\psi_1}$}
\rput(0.6,5.65){$\ket{\psi_2}$}

\rput{110}(5,4){
\psline[linewidth=0.04cm,arrowsize=0.1cm 2.0,arrowlength=1.5,arrowinset=0.4]{<->}(-1.0,0)(1.0,0)
\psdot(-1.0,0)
\psdot(1.0,0)}
\rput(4.6,5.1){$\rho_1$}
\rput(5.4,2.85){$\rho_2$}

\rput{-139}(5,0){
\psline[linewidth=0.04cm,arrowsize=0.1cm 2.0,arrowlength=1.5,arrowinset=0.4]{->}(0,0)(0.5,1)
\psdot(0.5,1)
\psline[linewidth=0.04cm,arrowsize=0.1cm 2.0,arrowlength=1.5,arrowinset=0.4]{->}(0,0)(0.5,0)
\psdot(0.5,0)
\psdot(0.5,0.5)
}

\end{pspicture} 
}

\caption{\selectlanguage{english}%
\label{fig:unital}\foreignlanguage{american}{Quantum channels on
a single qubit: The upper figure shows a unital quantum channel $\Lambda_{u}$
(green arrow) which maps the maximally mixed state $\frac{1}{2}\openone$
onto itself: $\Lambda_{u}\left(\frac{1}{2}\openone\right)=\frac{1}{2}\openone$.
Two orthogonal states $\ket{\psi_{1}}$ and $\ket{\psi_{2}}$ with
collinear Bloch vectors are mapped onto the states $\rho_{1}=\Lambda_{u}\left(\ket{\psi_{1}}\bra{\psi_{1}}\right)$
and $\rho_{2}=\Lambda_{u}\left(\ket{\psi_{2}}\bra{\psi_{2}}\right)$
with collinear Bloch vectors. The lower figure shows a non-unital
quantum channel $\Lambda_{nu}$ (yellow arrow) which maps the maximally
mixed state onto the state $\sigma=\Lambda_{nu}\left(\frac{1}{2}\openone\right)\neq\frac{1}{2}\openone$.
The Bloch vectors of $\sigma_{1}=\Lambda_{nu}\left(\ket{\psi_{1}}\bra{\psi_{1}}\right)$
and $\sigma_{2}=\Lambda_{nu}\left(\ket{\psi_{2}}\bra{\psi_{2}}\right)$
add up to twice the non-zero Bloch vector of $\sigma$, see main text.}\selectlanguage{american}
}
\end{figure}

Before presenting the main result of this Letter, we introduce the
\emph{semi-classical} channel $\Lambda_{sc}$. It maps all input states
$\rho$ onto states $\Lambda_{sc}\left(\rho\right)$ which are diagonal
in the same basis: 
\begin{equation}
\Lambda_{sc}\left(\rho\right)=\sum_{k}p_{k}\left(\rho\right)\ket{k}\bra{k}.\label{eq:sc}
\end{equation}
 The nonnegative probabilities $p_{k}\left(\rho\right)$ can in general
depend on the input state $\rho$, while the orthogonal states $\ket{k}$
are independent of $\rho$. Such a channel is e.g. realized by complete
decoherence, after which only the diagonal elements of a density matrix
may be non-zero. We are now in the position to prove the following
theorem. 
\begin{thm}
\label{thm:1}A local quantum channel acting on a single qubit can
create quantum correlations in a multi-qubit system if and only if
it is neither semi-classical nor unital.
\end{thm}
\begin{proof} For simplicity we restrict ourselves to two qubits
only. A generalization to an arbitrary number of qubits is straightforward.
The action of a local semi-classical channel $\Lambda_{sc}^{A}$ on
the classically correlated state (\ref{eq:cc}) is, due to linearity,
\begin{equation}
\Lambda_{sc}^{A}\left(\rho_{cc}\right)=\sum_{i,j}p_{ij}\Lambda_{sc}^{A}\left(\ket{i^{A}}\bra{i^{A}}\right)\otimes\ket{j^{B}}\bra{j^{B}}.
\end{equation}
 The definition of a semi-classical channel in Eq. (\ref{eq:sc})
directly implies that $\Lambda_{sc}^{A}\left(\rho_{cc}\right)$ is
classically correlated.

Now we will show that a local unital channel never creates quantum
correlations in a multi-qubit system. A local unital channel $\Lambda_{u}^{A}$
on the qubit $A$ takes a classically correlated state to the state
\begin{equation}
\Lambda_{u}^{A}\left(\rho_{cc}\right)=\sum_{i,j}p_{ij}\Lambda_{u}^{A}\left(\ket{i^{A}}\bra{i^{A}}\right)\otimes\ket{j^{B}}\bra{j^{B}}.
\end{equation}
 The action of the unital channel on the pure state $\ket{i^{A}}\bra{i^{A}}$
can be studied using the Bloch representation: $\ket{0^{A}}\bra{0^{A}}=\frac{1}{2}\left(\openone_{A}+\sum_{i}r_{i}\sigma_{i}^{A}\right)$,
where $\sigma_{i}^{A}$ are the Pauli operators with $i\in\{x,y,z\}$,
and $\ket{1^{A}}\bra{1^{A}}=\frac{1}{2}\left(\openone_{A}-\sum_{i}r_{i}\sigma_{i}^{A}\right)$.
Using linearity and unitality of $\Lambda_{u}^{A}$ we see that the
state $\ket{0^{A}}\bra{0^{A}}$ is mapped onto the state $\rho_{0}^{A}=\Lambda_{u}^{A}\left(\ket{0^{A}}\bra{0^{A}}\right)=\frac{1}{2}\left(\openone_{A}+\sum_{i}r_{i}\Lambda_{u}^{A}\left(\sigma_{i}^{A}\right)\right)$.
The same procedure for $\ket{1^{A}}\bra{1^{A}}$ results in $\rho_{1}^{A}=\Lambda_{u}^{A}\left(\ket{1^{A}}\bra{1^{A}}\right)=\frac{1}{2}\left(\openone_{A}-\sum_{i}r_{i}\Lambda_{u}^{A}\left(\sigma_{i}^{A}\right)\right)$.
Note that the Bloch vectors of the states $\rho_{0}^{A}$ and $\rho_{1}^{A}$
point into opposite directions, see Fig. \ref{fig:unital} for illustration.
States with this property can be diagonalized in the same basis. This
implies that $ $it is possible to write the state $\Lambda_{u}^{A}\left(\rho_{cc}\right)$
in the form (\ref{eq:cc}). Thus we proved that local unital quantum
channels cannot create quantum correlations in a classically correlated
multi-qubit state.

In the following we will complete the proof of Theorem \ref{thm:1}
by showing that any local quantum channel $\Lambda_{nu}^{A}$ that
is neither unital nor semi-classical can create quantum correlations.
By definition $\Lambda_{nu}^{A}$ maps the maximally mixed state $\frac{1}{2}\openone_{A}$
onto some state that is not maximally mixed: 
\begin{equation}
\Lambda_{nu}^{A}\left(\frac{1}{2}\openone_{A}\right)=\frac{1}{2}\left(\openone_{A}+\sum_{i}s_{i}\sigma_{i}^{A}\right),
\end{equation}
 with $\sum_{i}s_{i}^{2}\neq0$. Since we demand that the quantum
channel is not semi-classical, there exists a state $\ket{\psi^{A}}$
such that $\Lambda_{nu}^{A}\left(\ket{\psi^{A}}\bra{\psi^{A}}\right)$
is not diagonal in the eigenbasis of $\Lambda_{nu}^{A}\left(\frac{1}{2}\openone_{A}\right)$.
Again we consider the Bloch representation 
\begin{equation}
\Lambda_{nu}^{A}\left(\ket{\psi^{A}}\bra{\psi^{A}}\right)=\frac{1}{2}\left(\openone_{A}+\sum_{j}r_{j}\sigma_{j}^{A}\right)
\end{equation}
and note that the two Bloch vectors $\boldsymbol{r}$ and $\boldsymbol{s}$
are linearly independent. Otherwise the states $\Lambda_{nu}^{A}\left(\ket{\psi^{A}}\bra{\psi^{A}}\right)$
and $\Lambda_{nu}^{A}\left(\frac{1}{2}\openone_{A}\right)$ could
be diagonalized in the same basis, which is in contradiction to the
definition of $\ket{\psi^{A}}$. Consider now the classically correlated
state 
\begin{equation}
\rho_{cc}=\frac{1}{2}\ket{\psi^{A}}\bra{\psi^{A}}\otimes\ket{0^{B}}\bra{0^{B}}+\frac{1}{2}\ket{\phi^{A}}\bra{\phi^{A}}\otimes\ket{1^{B}}\bra{1^{B}}
\end{equation}
 with orthogonal states $\braket{\psi^{A}|\phi^{A}}=0$. We can write
the states as $\ket{\psi^{A}}\bra{\psi^{A}}=\frac{1}{2}\left(\openone_{A}+\sum_{i}v_{i}\sigma_{i}^{A}\right)$,
and $\ket{\phi^{A}}\bra{\phi^{A}}=\frac{1}{2}\left(\openone-\sum_{i}v_{i}\sigma_{i}^{A}\right)$.
We define the vector $\boldsymbol{w}$ such that the equality $\Lambda_{nu}\left(\sum_{i}v_{i}\sigma_{i}^{A}\right)=\sum_{i}w_{i}\sigma_{i}^{A}$
with $\sum_{i}w_{i}^{2}\neq0$ is satisfied. This is always possible,
since $\Lambda_{nu}$ is trace-preserving. The action of the channel
onto the two states $\ket{\psi^{A}}$ and $\ket{\phi^{A}}$ is as
follows: 
\begin{align}
\Lambda_{nu}^{A}\left(\ket{\psi^{A}}\bra{\psi^{A}}\right) & =\frac{1}{2}\left(\openone_{A}+\sum_{i}\left(s_{i}+w_{i}\right)\sigma_{i}^{A}\right),\\
\Lambda_{nu}^{A}\left(\ket{\phi^{A}}\bra{\phi^{A}}\right) & =\frac{1}{2}\left(\openone_{A}+\sum_{i}\left(s_{i}-w_{i}\right)\sigma_{i}^{A}\right).
\end{align}
 As noted above, the two Bloch vectors $\boldsymbol{s}$ and $\boldsymbol{r}=\boldsymbol{s}+\boldsymbol{w}$
are linearly independent. The same must hold for the vectors $\boldsymbol{s}+\boldsymbol{w}$
and $\boldsymbol{s}-\boldsymbol{w}$. This implies that the two states
$\Lambda_{nu}^{A}\left(\ket{\psi^{A}}\bra{\psi^{A}}\right)$ and $\Lambda_{nu}^{A}\left(\ket{\phi^{A}}\bra{\phi^{A}}\right)$
are not diagonal in the same basis. This completes the proof. \end{proof}

So far we saw that local unital and local semi-classical channels
acting on a single qubit cannot create quantum correlations from a
classically correlated multi-qubit state. These results hold independently
of the chosen measure for quantum correlations. In the following we
will go one step further by showing that these local channels never
increase a very general class of measures for quantum correlations
in multi-qubit systems. We consider distance-based measures of quantum
correlations $Q_{D}$, which are defined via the minimal distance
D to the set of the classically correlated states $CC$ \cite{Modi2010,Daki'c2010},
\begin{equation}
Q_{D}=\min_{\sigma\in CC}D\left(\rho,\sigma\right),\label{eq:QD}
\end{equation}
where $D$ does not necessarily have to be a distance in the mathematical
sense. The statement mentioned above will be shown to hold for all
distance measures $D$ with the property of being non-increasing under
any quantum channel $\Lambda$, i.e. 
\begin{equation}
D\left(\Lambda\left(\rho\right),\Lambda\left(\sigma\right)\right)\leq D\left(\rho,\sigma\right).\label{eq:D}
\end{equation}
 This property is also frequently used for defining entanglement measures
\cite{Vedral1997,Vedral1998}. \begin{thm} \label{thm:2}Quantum
correlations $Q_{D}$ in multi-qubit systems do not increase under
local unital channels $\Lambda_{lu}$ and local semi-classical channels
$\Lambda_{lsc}$: 
\begin{align}
Q_{D}\left(\Lambda_{lu}\left(\rho\right)\right) & \leq Q_{D}\left(\rho\right),\label{eq:lu}\\
Q_{D}\left(\Lambda_{lsc}\left(\rho\right)\right) & \leq Q_{D}\left(\rho\right).\label{eq:lsc}
\end{align}
 \end{thm} 
\begin{proof}
Let $\xi$ be the classically correlated state which minimizes the
distance, i.e. $Q_{D}\left(\rho\right)=D\left(\rho,\xi\right)$. Using
the property (\ref{eq:D}) of the distance to be nonincreasing under
quantum channels we obtain 
\begin{align}
Q_{D}\left(\rho\right) & =D\left(\rho,\xi\right)\geq D\left(\Lambda_{lu}\left(\rho\right),\Lambda_{lu}\left(\xi\right)\right),\\
Q_{D}\left(\rho\right) & =D\left(\rho,\xi\right)\geq D\left(\Lambda_{lsc}\left(\rho\right),\Lambda_{lsc}\left(\xi\right)\right).
\end{align}
 Now we use Theorem \ref{thm:1} noting that local unital channels
$\Lambda_{lu}$ and local semi-classical channels $\Lambda_{lsc}$
map the classically correlated state $\xi$ onto another classically
correlated state $\Lambda\left(\xi\right)$ which is not necessarily
the one that minimizes the distance to $\Lambda\left(\rho\right)$.
This observation finishes the proof. 
\end{proof}
One example for a measure that satisfies the properties (\ref{eq:lu})
and (\ref{eq:lsc}) - and thus Theorem \ref{thm:2} holds - is the
{\em geometric measure of quantumness} which we define as 
\begin{equation}
Q_{G}\left(\rho\right)=\min_{\sigma\in CC}\left(1-F\left(\rho,\sigma\right)\right)\label{eq:QEg-1}
\end{equation}
with the fidelity $F\left(\rho,\sigma\right)=\left(\mathrm{Tr}\sqrt{\sqrt{\rho}\sigma\sqrt{\rho}}\right)^{2}$.
Using the fact that the fidelity is non-decreasing on quantum channels
together with Theorem \ref{thm:2}, we see that the geometric measure
of quantumness does not increase under local unital channels and local
semi-classical channels. Alternatively, we can use the quantum relative
entropy $S\left(\rho||\sigma\right)=-\mathrm{Tr}\left[\rho\log_{2}\sigma\right]+\mathrm{Tr}\left[\rho\log_{2}\rho\right]$,
which also fulfills the property (\ref{eq:D}) \cite{Vedral1997,Vedral1998}.
From Theorem \ref{thm:2} follows that the resulting measure of quantum
correlations $Q_{S}=\min_{\sigma\in CC}S\left(\rho||\sigma\right)$
does not increase under local unital and local semi-classical channels.
$Q_{S}$ was also studied in \cite{Piani2011}, where it was called
relative entropy of quantumness.

So far we considered states consisting of an arbitrary number of qubits.
We have shown that local unital and local semi-classical channels
acting on a single qubit never increase quantum correlations as defined
by $Q_{D}$ \texttt{in} Eq. (\ref{eq:QD}). On the other hand, any
local channel which is non-unital and not semi-classical can in principle
create quantum correlations, independently of the considered measure,
out of a classically correlated state. An example for such a channel
is the amplitude damping channel as a model for dissipation. Thus,
dissipation can increase quantum correlations.

At the present stage it is natural to ask the question, for what kind
of input states this behavior can or cannot be observed in general.
The following theorem shows that pure states are special. \begin{thm}
The geometric measure of quantumness of multipartite systems with
arbitrary dimension cannot increase under any local quantum channel,
if the initial state is pure: 
\begin{equation}
Q_{G}\left(\Lambda_{l}\left(\ket{\psi}\bra{\psi}\right)\right)\leq Q_{G}\left(\ket{\psi}\bra{\psi}\right),
\end{equation}
 where $\Lambda_{l}$ is an arbitrary local quantum channel.\end{thm}
\begin{proof} Let $\xi\in CC$ be defined such that $Q_{G}\left(\ket{\psi}\bra{\psi}\right)=1-F\left(\ket{\psi}\bra{\psi},\xi\right)$.
Using the properties of the fidelity $F$ we see that $\xi$ can be
chosen to be a pure product state $\xi=\ket{\phi}\bra{\phi}$. Moreover
$1-F$ does not increase under the action of any quantum channel,
i.e. $1-F\left(\ket{\psi}\bra{\psi},\ket{\phi}\bra{\phi}\right)\geq1-F\left(\Lambda_{l}\left(\ket{\psi}\bra{\psi}\right),\Lambda_{l}\left(\ket{\phi}\bra{\phi}\right)\right)$.
Since $\ket{\phi}$ is a product state, $\Lambda_{l}\left(\ket{\phi}\bra{\phi}\right)$
is also a product state. This observations completes the proof. \end{proof}
Note that Theorem 3 does {\em not} follow from the fact that for
pure states the amount of quantum correlations is equal to the amount
of entanglement.

So far we have shown that quantum correlations in multi-qubit systems
cannot increase under local unital quantum channels. A prominent example
for a unital channel is the phase damping channel, which is a model
for decoherence in a quantum system. Under decoherence the quantum
state $\rho=\sum_{i,j}\rho_{ij}\ket{i}\bra{j}$ is transformed to
the state 
\begin{equation}
\Lambda\left(\rho\right)=\sum_{i}\rho_{ii}\ket{i}\bra{i}+\left(1-p\right)\sum_{i\neq j}\rho_{ij}\ket{i}\bra{j}\label{eq:Lambda}
\end{equation}
 with the damping parameter $0\leq p\leq1$. Since $\Lambda$ is unital,
it is not possible to create quantum correlations with local phase
damping in a multi-qubit system. Surprisingly, this is not true if
the local systems are not qubits: qubits are special. This can be
demonstrated via the classically correlated state as input: 
\begin{equation}
\rho_{cc}=\frac{1}{2}\ket{\psi^{A}}\bra{\psi^{A}}\otimes\ket{0^{B}}\bra{0^{B}}+\frac{1}{2}\ket{\phi^{A}}\bra{\phi^{A}}\otimes\ket{1^{B}}\bra{1^{B}}\label{eq:rhocc}
\end{equation}
 with the orthogonal single-qutrit states $\ket{\psi^{A}}=\frac{1}{\sqrt{3}}\left(-\ket{0^{A}}+\ket{1^{A}}+\ket{2^{A}}\right)$
and $\ket{\phi^{A}}=\frac{1}{\sqrt{2}}\left(\ket{0^{A}}+\ket{1^{A}}\right)$.
We will show that a local phase damping channel $\Lambda_{A}$ acting
on subsystem $A$ generates quantum correlations. We consider the
action of the channel (\ref{eq:Lambda}) with the damping parameter
$p=\frac{1}{2}$ on the state $\rho_{cc}$ in Eq. (\ref{eq:rhocc}):
\begin{eqnarray}
\Lambda_{A}\left(\rho_{cc}\right) & = & \frac{1}{2}\sum_{i=1}^{3}\lambda_{i}\ket{\psi_{i}^{A}}\bra{\psi_{i}^{A}}\otimes\ket{0^{B}}\bra{0^{B}}\nonumber \\
 &  & +\frac{1}{2}\sum_{j=1}^{3}\mu_{j}\ket{\phi_{j}^{A}}\bra{\phi_{j}^{A}}\otimes\ket{1^{B}}\bra{1^{B}},\label{nocc}
\end{eqnarray}
 where the states $\left\{ \ket{\psi_{i}^{A}}\right\} $ are the eigenstates
of $\Lambda_{A}\left(\ket{\psi^{A}}\bra{\psi^{A}}\right)$ with the
corresponding eigenvalues $\lambda_{i}$. Similarly the states $\left\{ \ket{\phi_{j}^{A}}\right\} $
are eigenstates of $\Lambda_{A}\left(\ket{\phi^{A}}\bra{\phi^{A}}\right)$
with the eigenvalues $\mu_{j}$. One can see as follows that the state
$\Lambda_{A}\left(\rho_{cc}\right)$ is quantum correlated: The eigenvalues
of $\Lambda_{A}\left(\ket{\psi^{A}}\bra{\psi^{A}}\right)$ are given
by $\lambda_{1}=\frac{2}{3}$, and $\lambda_{2}=\lambda_{3}=\frac{1}{6}$.
The eigenstate to the largest eigenvalue $\lambda_{1}$ is given by
$\ket{\psi_{1}^{A}}=\ket{\psi^{A}}$. It is easy to check that $\ket{\psi_{1}^{A}}$
is not an eigenstate of $\Lambda_{A}\left(\ket{\phi^{A}}\bra{\phi^{A}}\right)$,
and therefore the state in Eq. (\ref{nocc}) is not classically correlated.
Thus we proved that it is possible to create quantum correlations
with a local phase damping channel, i.e. via local decoherence.

In conclusion, we have investigated the effect of {\em local} noisy
channels (i.e. trace-preserving completely positive maps) on quantum
correlations. While entanglement can never increase under such local
channels, quantum correlations without entanglement may or may not
increase, depending on the type of channel and the type of input state.
For multi-qubit systems we fully answer the question which local channels
can increase quantum correlations: unital and semi-classical local
channels cannot enhance quantum correlations, while non-unital and
non-semi-classical local channels (e.g. dissipation, corresponding
to amplitude damping) can increase quantum correlations. Surprisingly,
for higher-dimensional systems, even unital channels such as decoherence,
corresponding to phase-damping, can generate quantum correlations
from an initially classically correlated state. However, quantum correlations
as quantified by the geometric measure of quantumness can become larger
under local channels only when the initial state is mixed. - Thus,
we have shed some light on the behavior of quantum correlated states
in a noisy environment.

We acknowledge partial financial support by Deutsche Forschungsgemeinschaft
(DFG) and by the ELES foundation.

\paragraph{Note added:}

While finishing this Letter we became aware of related work \cite{Ciccarello}.
There the authors show that the quantum discord can increase under
a local amplitude damping channel.

 \bibliographystyle{apsrev4-1}
\bibliography{literature}

\end{document}